\documentclass[12pt]{article}
\usepackage{amsmath,amssymb,amsfonts,amsthm}
\usepackage{color}
\usepackage[left=3cm,top=3cm,right=3cm,bottom=3cm, includefoot]{geometry}
\usepackage{graphicx}
\usepackage{tikz}
\usetikzlibrary{matrix,arrows}
\newtheorem{theorem}{Theorem}[section]
\newtheorem{lemma}{Lemma}[section]
\newtheorem{corollary}{Corollary}[section]
\newtheorem{definition}{Definition}[section]
\newtheorem{example}{Example}[section]

\numberwithin{equation}{section}

\begin{document}
\title{MacWilliams identities for  poset level weight enumerators of linear codes}

\date{}

\maketitle
\begin{center}

\author{{A.\ Seda}$^{1}$, S.\ Vedat $^{2,*}$\\
\textnormal{$^{1}$ Department of Mathematics, Yildiz Technical University\\
$^{2}$ Department of Mathematical Engineering, Yildiz Technical University\\
$^{*}$Email: vsiap@yildiz.edu.tr}}
\end{center}

\begin{abstract}

Codes over various metrics such as Rosenbloom-Tsfasman (RT), Lee, etc. have been considered. Recently, codes over poset metrics have been studied. Poset metric is a great generalization of many metrics especially the well-known ones such as the RT and the Hamming metrics. Poset metric can be realized on the channels with localized error occurrences.  It has  been shown that MacWilliams identities are not admissible for codes over poset metrics in general [Kim and Oh, 2005]. Lately, to overcome this problem some further studies on MacWilliams identities over poset metrics has been presented. In this paper, we introduce new poset level weight enumerators of linear codes over Frobenius commutative rings. We derive MacWilliams-type identities for each of the given enumerators which generalize in great deal the previous results discussed in the literature. Most of the weight enumerators in the literature such as Hamming, Rosenbloom-Tsfasman and complete $m$-spotty weight enumerators  follow as corollaries to these identities especially.



\end{abstract}

\section{ History and Motivation of the Problem}

       The metric that is used for detecting errors hence correcting them afterwards depends on the channel  transferring or storing digital data, and the media used where different scenarios take place and different type of errors occur. The most used and well-known metric is the Hamming metric. A different error type where the sender knows the possible error location but not the error value itself is introduced by Bassalygo, Gelfand, and Pinsker \cite{2} and later by Roth and Seroussi in  \cite{evet}. Some work on this direction with applications is pursued by the same authors in \cite{3}, and also by Ahlswede, Bassalygo, and Pinsker in \cite{1}, and by Larsson in \cite{4} and Roth in  \cite{evet}. A new construction method for codes correcting multiple localized burst errors is proposed in \cite{0}. Another work that distinguishes between errors by prioritizing some cases is presented in \cite{Packet} and the performance of such a scheme is verified on  the  popular H.264/AVC codec. For instance, in some wireless communication systems, headers of the transmitted data such as the frame control, duration and address are more important than the
         frame body. This is due the fact that the errors in the headers may cause a rejection of the transmission which makes these positions  more important than the others. In order to solve this problem a method called unequal error protection (UEP) is  first introduced by Masnik et al.  \cite{masnik}.

         Some recent studies towards this direction are pursued in \cite{chang,huang,morelos} and further constructions of such codes are presented and some bounds are studied by Kuriata in \cite{kuriata}.

        As seen above one can device a code with the knowledge that some bits due to their location in a codeword  could be more  vulnerable to the errors or play more important role comparing the other ones. To sort out this phenomena one may use a poset metric which depends on the position of bits and hence can help on distinguishing between the locations of bits. Further, burst errors are also very common and hence this suggests to combine these two properties and define a new metric which we call $m$-spotty poset metric and study the relation between the weight distributions of codes and their duals. 
        
  The motivation of this work mainly relies on the very recent studies done on both $m$-spotty and poset weight enumerators. The $m$-spotty weights were introduced since they are capable of modelling the errors that frequently appear in flash memory disks as bursts in bytes \cite{umanesan1,umanesan2}. For further information on spotty byte errors the readers can refer to \cite{fujiwara,suzuki1}.  Due to this important fact the researchers have been extensively studying its properties and especially MacWilliams identities that relate the weight enumerators of codes to their duals. In \cite{suzuki2},  $m$-spotty weight enumerators of $m$-spotty byte error control codes are introduced and the MacWilliams identity for m-spotty Hamming weight enumerators for binary m-spotty byte error control codes are established. In addition, this generalization  includes the MacWilliams identity for the Hamming weight enumerator as a special case. In \cite{irfansiap1},  $m$-spotty Lee weights are introduced and a MacWilliams-type identity for m-spotty Lee weight enumerators is proved .  In \cite{vedatsiap1}, the results obtained in  \cite{suzuki2} are extended further to arbitrary finite fields.  In \cite{irfansiap2},  $m$-spotty weights and m-spotty weight enumerator of linear codes over the ring $\mathbb{F}_2+u\mathbb{F}_2$ are introduced and a MacWilliams-type identity is established.
Later,  in \cite{vedatsiap2}, $m$-spotty weights for codes over the ring $\mathbb{F}_2+v\mathbb{F}_2=\{0,1,v,1+v \}$ with $ v^2=v$ are introduced and a MacWilliams-type identity is also proved. Recently, in \cite{vedatsiap3,vedatsiap4}, the studies on MacWilliams identities for m-spotty byte error control codes have been considered for different metrics. Further, in \cite{sharma1,sharma2,sharma3,suzuki3}, the studies on MacWilliams identities for m-spotty byte error control codes have been considered for some new m-spotty weight enumerators and their properties are studied.

  On the other hand, studies on poset weight enumerators are also very recent \cite{choi,felix}. In \cite{brualdi}, the original problem  studied by Niederreiter \cite{Niederreiter1,Niederreiter2,Niederreiter3} on optimal parameters of linear codes is considered in a  more general setting of partially ordered sets and in this setting  poset-codes are introduced. Niederreiter's setting was viewed as the disjoint union of chains and extended some of Niederreiter's bounds and also obtained bounds for posets which are the product of two chains.
  Lately, poset codes are shown to outperform better then the classical ones while applied in decoding processes \cite{posetdecoders}. In \cite{kwang}, all poset structures that admit the MacWilliams identity with respect to ideal based weights are classified, and the MacWilliams identities for poset weight enumerators corresponding to such posets are derived. It is shown that being a hierarchical poset is a necessary and sufficient condition for a poset to admit such a  MacWilliams type identity \cite{kwang}. An explicit relation is also derived between the poset weight distribution of a hierarchical poset code and the $\bar{P}$ (dual poset)-weight distribution of the dual code \cite{kwang}. Recently, in \cite{akbiyik1}, an alternative P-complete weight enumerator of linear codes with respect to poset metric that includes the hierarchical posets consisting of more variables is defined and a MacWilliams-type identity is proved. Poset weights also generalize the Hamming weights and Rosenbloom-Tsfasman weights. The interesting case is the Rosenbloom-Tsfasman (RT) which is a special poset consisting of a single chain. Some work on RT metric over various special finite Frobenius rings  related to MacWilliams identities is done \cite{irfansiap5,irfansiap4, irfansiap3}.

In Section \ref{sec2}, some facts and  preliminaries that will be referred to in the sequel are presented. In  Section \ref{sec3}, the byte weight enumerator for a linear poset code over Frobenius ring is introduced and  a MacWilliams-type identity is also proven. In Section \ref{sec4}, the definition of complete m-spotty poset level weight enumerator for a poset code ($P-$ code) $C$ over Frobenius rings is introduced and a MacWilliams-type identity for complete m-spotty poset level weight enumerator is proved. Moreover, the definition of poset level weight enumerator of binary linear codes is extended to linear codes (\cite{akbiyik1}) over Frobenius rings and  a MacWilliams-type identity for these weight enumerators is obtained. Also, a new m-spotty weight enumerator which is called m-spotty poset level weight enumerator is introduced and  the MacWilliams identity with respect to this weight enumerator is also established. In Section \ref{sec5}, some illustrative examples are given and Section \ref{sec6} concludes the paper.

\section{Preliminiaries}\label{sec2}

In this section, we state  some basic results and definitions.
For some terms and detailed information especially regarding Frobenius rings,
the readers are welcome to refer to \cite{wood}.

Let $R$ be a finite ring and let $N$ be a positive integer. A linear code $C$ over $R$ is an $R-$submodule of $R^N$. The elements of $C$ are referred as codewords. By abuse of terminology, the
 elements of $R^N$ will be called vectors.
\begin{definition}\label{byteweightdef1} The $i$th bytes $u^i$ of a vector $u\in R^N$ whose index set is partitioned into $s$ not necessarily equal parts each of size $n_i$ is defined by
\begin{eqnarray*}
u^1&=&(u_{11},u_{12},...,u_{1{n_1}})\in R^{n_1},\\
u^2&=&(u_{21},u_{22},...,u_{2{n_2}})\in R^{n_2},\\
&...&\\
u^s&=&(u_{s1},u_{s2},...,u_{s{n_s}})\in R^{n_s}.
\end{eqnarray*}
where $N=n_1+n_2+...+n_s.$ 
\end{definition}
The $s-$tuple $(u^1,u^2,...,u^s)$ is called the $s-$level representation of a vector $u$ of length $N$ such that each $n_i-$tuple $u^i$ denotes the part in the $i$th level of the vector.

Throughout this paper $N$, $u^i$'s and $v^i$'s will be used as in Definition \ref{byteweightdef1}.

\begin{example} \label{ex31}
Let $R=\mathbb{F}_2$ and $N=6$ and $u=(1,0,1,1,0,0)\in \mathbb{F}_2^6$. Then, $u^1=(1,0)$, $u^2=(1)$,$u^3=(1,0,0)$ with respect to the given poset in Figure \ref{Fig100} and the $3-$level $\left( {\left\{ {1,2} \right\} < \left\{ 3 \right\} < \left\{ {4,5,6} \right\}} \right)$  representation of $u$ is given as $(10,1,100)$.
\begin{figure}
\centering
\hfill
   \includegraphics[bb=0 0 300 150 ,scale=0.9] {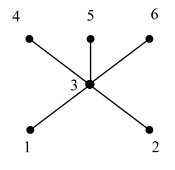}\\
\caption{A poset of size 6 and 3 levels }\label{Fig100}
\end{figure}
\end{example}

For any two vectors $u = \left( {{u^1},{u^2},...,{u^s}} \right)$, $v =
\left( {{v^1},{v^2},...,{v^s}} \right) \in {R^N}$, the inner
product of $u$ and $v$ is given by
\begin{equation} \label{2.1}
\left\langle {u,v} \right\rangle  = \sum\limits_{i = 1}^s
{\left\langle {{u^i},{v^i}} \right\rangle = \sum\limits_{i = 1}^s(\sum_{j=1}^{n_i}
{u_{ij}v_{ij}})}.
\end{equation}

Given a linear code $C\subset R^N$ we define its dual code with respect to the inner product in (\ref{2.1}) as

\begin{equation}
{C^ \bot } = \left\{ {v \in {R^{N}}:\left\langle {u,v}
\right\rangle  = 0\,\, for\,\, all\,\, u \in C} \right\}.
\end{equation}

The Hamming weight of a vector $v\in R^N$, denoted by $w(v)$, is the number of non-zero coordinates of $v$. The  Hamming distance, a metric on $R^N$, between  two vectors $u$ and $v$ is $d(u,v)=w(u-v)$. The minimum distance of the linear code $C$ is the minimal Hamming distance between any two distinct codewords of $C.$ The smallest nonzero Hamming weight in a code is called the minimum Hamming weight of the code. In linear code case the minimum distance and the minimum Hamming weight are equal. A different metric which is refereed to as the poset metric has been of interest to the algebraic coding theorists pretty recently. This is a position based metric and it is a generalization of some important metrics such as Hamming and Rosenbloom-Tsfasman metrics. Now we present some basics on the poset metric over $R^N$. Suppose that $(P,\leq)$ is a partially ordered set of size $N$.  For all $x\in I$ and $y\leq x$ if $y\in I,$  then this subset $I$ of $P$ is called \emph{an ideal} of $P$. If $A\subset P$, then $\langle A\rangle$ is the smallest ideal of $P$ containing $A$. Suppose that $P=\{1,2,3,...,N\}$ and the coordinate positions of elements of ${R}^N$ are labeled by the elements of $P$. For any vector $ v\in {R}^N,$ the $P-$ weight of $v$ is defined by $w_P(v)=| \langle supp(v)\rangle |$ which is the size of  the smallest ideal of $P$ containing the support of $v$, where $supp(v)=\{i\in P: v_i\neq 0\}$. Then naturally the $P-$distance, a metric on $R^N,$ between two vectors $u$ and $v$ is defined as $d_P(u,v)=w_P(u-v)$. There are two direct observations regarding the poset weight: if $P$ is  antichain, then the $P-$ weight is the same as  Hamming weight.  If $P$ consists of a single chain, then $P-$ weight coincides with Rosenbloom-Tsfasman (RT)  weight. If $R^N$ is endowed with a poset metric, then we call a subset $C$ of $R^N$ a poset code. If the poset metric is derived from a poset $P$, then $C$ is called a $P-$code.

\begin {definition} \cite{article. 7,article.2} Suppose that $C$ is a linear P-code of length $N$. $$W_{C,P}(x)=\sum_{u\in C}x^{w_{P}(u)}=\sum_{i=0}^N {A_{i,P}} x^i$$
is called the poset weight enumerator of $C$  where $A_{i,P}=|\{u\in C \mid w_P(u)=i\}|$.
\end {definition}

The following example, presented in \cite{article.2}, shows that a direct attempt for obtaining a MacWilliams identity for poset codes is not possible in general.

\begin {example} \label{ex32}\cite{article.2} Let $P=\{1,2,3\}$ be a poset with order relation $1<2<3,$ i.e., a single chain poset. Let $C_1=\{000,001\}$ and $C_2=\{000,111\}$ be two binary linear codes with respect to the poset $P.$ The poset weight enumerators of $C_1$ and $C_2$ are given by  $W_{C_1,P}(x)=1+x^3=W_{C_2,P}(x)$. The dual codes of $C_1$ and $C_2$ are ${C_1}^{\perp}=\{000,100,010,110\}$ and ${C_2}^{\perp}=\{000,110,101,011\}$, respectively. The P- weight enumerators of the dual codes are given by $W_{{C_1}^\perp ,P}(x)=1+x+2x^2$ and $ W_{{C_2}^\perp,P}(x)=1+x^2+2x^3$ .
\end {example}

As seen in the example above, in \cite{article.2, hier} the problem to obtain a MacWilliams identity is achieved by restricting the structure of the poset to being a hierarchical poset and considering the dual code over the dual poset. In \cite{akbiyik1}, a complete poset weight enumerator is introduced and a MacWilliams identity is obtained for a broaden class of posets including hierarchical posets.  Recently, in \cite{akbiyik2}, the results obtained in  \cite{akbiyik1} has been generalized further to the posets that are forests and further both the code and its dual are considered over the same poset.

From now on we assume that all rings are finite commutative
Frobenius rings. We first give the basic definitions and
theorems, next we state and prove  MacWilliams-type identities of new weight enumerators for linear codes over
finite commutative Frobenius rings. 

\begin{definition}\label{def1} A finite ring $R$ is said to be a Frobenius ring
 if $\widehat R  \cong {R_R}.$
 Here, the character group of the additive group of $R$ is denoted
 by $\widehat R = Ho{m_{Z}}\left( {R,{{C}^ \times }}
 \right)$.
 \end{definition}

It is well-known (cf. \cite{wood}) that if $R$ is a finite
Frobenius ring, then $R$ and $\widehat R$ are isomorphic also as
left $R-$modules.

 \begin{example}\cite{wood}\label{examp1} The ring of integers
 modulo $m$ (${Z}_m$), Galois fields and rings, and $Ma{t_{n \times n}}\left( R \right)$, the ring of all
 $n \times n-$matrices over $R$ are examples of Frobenius rings.
 \end{example}

\begin{definition}\cite{wood}\label{def2} A character $\zeta$ of $R$ is
a generating character if the mapping

\begin{equation}\label{characterfunction}
\zeta :R \to \widehat R,\,\,\zeta \left( u \right) =\chi_u({v})= \chi \left(
{\left\langle {u,v} \right\rangle} \right)
\end{equation}
is an isomorphism of $R-$modules for all $u,v \in R,$ and for all
$\chi \in \widehat R.$
\end{definition}

\begin{theorem}\cite{Claasen}\label{theo1} Let $\chi$ be a character of $R.$ Then
$\chi$ is a generating character if and only if the kernel of
$\chi$ contains no non-zero ideals.
\end{theorem}

From Definition \ref{def1}, a finite ring is Frobenius if and only
if it admits a generating character.

\begin{lemma}\cite{macwilliams}\label{lemma1} Let $I \ne \left\{ 0 \right\}$ be an ideal of
$R$ and $\chi$ be a generating character of $R$. Then,
\begin{equation}
\sum\limits_{a \in I} {\chi \left( a \right)}  = 0.
\end{equation}
\end{lemma}

By Lemma \ref{lemma1} and $\chi \left( 0 \right) = 1$, we obtain
the following corollary:

\begin{corollary}\label{corol1} Suppose that $R$ is a ring, with a generating
character $\chi.$ Then,

\begin{equation}
\sum\limits_{a \in I\backslash \left\{ 0 \right\}} {\chi \left( a
\right)}  =  - 1.
\end{equation}
\end{corollary}

The following lemma plays an important role in deriving a
MacWilliams-type identity for weight enumerators over
finite commutative Frobenius rings:

\begin{lemma}\label{lemma2} Let $f$ be a function defined on $R^N$, and let $\chi$
be a generating character on $R.$ The Hadamard transform
$\widetilde f$ of $f$ is defined by

\begin{equation}
\widetilde f\left( u \right) = \sum\limits_{v \in {R^{N}}} {\chi
\left( {\left\langle {u,v} \right\rangle } \right)f\left( v
\right)} ,\,\,\,u \in {R^{N}}.
\end{equation}

Then, the following relation holds between $f(v)$ and $\widetilde
f\left( u \right):$

\begin{equation}\label{eq81}
\sum\limits_{v \in {C^ \bot }} {f\left( v \right) =
\frac{1}{{\left| C \right|}}\sum\limits_{u \in C} {\widetilde
f\left( u \right)} },
\end{equation}
where $\left| C \right|$ denotes the size of the code $C$.
\end{lemma}

\textbf{\textbf{Proof.}} Proof is similar to that of Lemma 2 of
\cite{macwilliams}.

\section{Byte poset level weight enumerator of a linear $P_R-$code}\label{sec3}

  Byte weight enumerators are  introduced for Hamming metrics due to the burst errors that occur in both  transmission  or storing  data processes. In \cite{simonis}, codewords with an index partition are introduced and MacWilliams identity is proven. In this section we combine the byte-weight concept with posets. Further, we point out that the definition in \cite{simonis} becomes a special case by choosing a special poset.

Now we introduce the following weight function that is going to appear in the proof of Theorem \ref{byteweighttheo1}.

\begin{definition} \label{byteweightdef2}
Let $C$ be a linear $P_R-$code over a poset of size $N$ and $s$ levels and let $R=\{\beta_0,\beta_1,...,\beta_{q-1}\}$. To each $n_i$-tuple $u^i$ in $R^{n_i}$ we define the weight function as follows:
\begin{itemize}
\item[(i)] if $S\neq k$, then $\eta_{S:i_1i_2...i_{n_k}}(\beta_{j_1},...,\beta_{j_{n_k}}) =0$,

\item[(ii)] if $S=k$, then
$\eta_{S:i_1i_2...i_{n_k}}(\beta_{j_1},...,\beta_{j_{n_k}}) =\left\{ \begin{array}{ll}
1, & \textrm{ if $(i_1,i_2,...,i_{n_k})=(j_1,j_2,...,j_{n_k})$ }\\
0, & \textrm{otherwise}\\
\end{array}\right.$
\end{itemize}
where $S,k \in \{1,2,...,s\}.$
\end{definition}

\begin{example}
Let $C$ be a binary $P_{\mathbb{F}_2}-$code with $N=6$ and $s=3$. Consider the vector $u$ in  Example \ref{ex31}. Then, $\eta_{1:10}(\beta_1,\beta_0)=\eta_{2:1}(\beta_1)=\eta_{3:100}(\beta_1,\beta_0,\beta_0)=1$. Otherwise
$\eta_{S:i_1i_2...i_{n_k}}(\beta_{j_1},...,\beta_{j_{n_k}})=0$.

\end{example}

We are now ready to give the definition of a byte poset level weight enumerator for a linear $P_R-$code.

\begin{definition}\label{byteweightdef3}
Let $C$ be a linear $P_R-$code over a poset of size $N$ and $s$ levels. Then the byte poset level weight enumerator of $C$ is defined as follows:
\begin{equation*}
B_W\left(C| z_{i: \overline{a}}  :   \overline{a} \in R^{n_i}, 1\leq i\leq s \right )  =\sum_{u\in C}\prod _{{n_j}=a}^{b}\prod _{S=1}^{s}z_{S:00...0}^{\mu_{S:00...0}(u)}...z_{S:i_1...i_{n_j}}^{{\mu_{S:i_1...i_{n_j}}(u)}}...z_{S:q-1...q-1}^{\mu_{S:q-1...q-1}(u)},
\end{equation*}
where $a=min\{n_j:j=1,...,s\}$, $b=max\{n_j:j=1,...,s\}$, and

\begin{equation}\label{mufunction}
\mu_{S:i_1...i_{n_j}}(u)=\sum_{k=1}^s \eta_{S:i_1i_2...i_{n_j}}(u^k).
\end{equation}
\end{definition}

Now, we state the MacWilliams identity for a linear $P_R-$code.

\begin{theorem}\label{byteweighttheo1}
Let $C$ be a linear $P_R-$code over a poset of size $N$ and $s$ levels. Then, the relation between the byte poset level weight enumerator of $C$ and its dual is given by
\begin{equation*}
B_W({C^\perp}| z_{i: \overline{a}}  :   \overline{a} \in R^{n_i}, 1\leq i\leq s)
=\frac{1}{|C|}\sum_{u\in C}\prod\limits_{{n_j} = a}^b {\prod\limits_{S = 1}^s {\prod\limits_{{i_1},{i_2},...,{i_{{n_j}}} \in Q}^{} A } },
\end{equation*}
where $\chi$ is a nontrivial additive character of $R$,
\[A = {\left( {\sum\limits_{{\beta _{{j_1}}},{\beta _{{j_2}}},...,{\beta _{{j_{{n_j}}}}} \in {R^{{n_j}}}}^{} {\chi \left( {\left( {{\beta _{{j_1}}},{\beta _{{j_2}}},...,{\beta _{{j_{{n_j}}}}}} \right)\left( {{\beta _{{i_1}}},{\beta _{{i_2}}},...,{\beta _{{i_{{n_j}}}}}} \right)} \right){z_{S:{j_1},{j_2},...,{j_{{n_j}}}}}} } \right)^{{\mu _{S:{i_1},{i_2},...,{i_{{n_j}}}}}\left( u \right)}},\]
and $Q=\{0,1,...,q-1\}$.
\end{theorem}
\begin{proof}
We recall Lemma \ref{lemma2} that
\begin{equation}\label{unluteorem}
\sum_{v\in C^\perp}f(v)=\frac{1}{|C|}\sum_{u\in C}\tilde{f}(u),
\end{equation}
 where
\begin{equation*}
\tilde{f}(u)=\sum_{v\in V}\chi_u(v)f(v).
\end{equation*}
Let $V=R^N$ and for $v\in R^N$ we take
\begin{equation*}
f(v)=\prod_{{n_j}=\alpha_{\min}}^{\alpha_{\max}}\prod _{S=1}^{s}z_{S:00...0}^{\mu_{S:00...0}(v)}...z_{S:i_1...i_{n_j}}   ^{\mu_{S:i_1...i_{n_j}}(v)}...z_{S:q-1...q-1}^{\mu_{S:q-1...q-1}}(v).
\end{equation*}
Then,
\begin{eqnarray*}
\tilde{f}(u)&=&\sum_{v\in R^N}\chi_u(v)f(v)\\
&=&\sum_{v\in R^N}\chi_u(v)\prod_{{n_j}=\alpha_{\min}}^{\alpha_{\max}}\prod _{S=1}^{s}z_{S:00...0}^{\mu_{S:00...0}(v)}...z_{S:i_1...i_{n_j}}^{\mu_{S:i_1...i_{n_j}}(v)}...z_{S:q-1...q-1}^{\mu_{S:q-1...q-1}(v)}.
\end{eqnarray*}

For fixed $u\in C$, first we compute $\tilde f(u)$. Using the definition of ${\mu}_{S:i_1,i_2,...,i_{n_j}}$ in (\ref{mufunction}), we obtain
\begin{equation*}
\tilde{f}(u)=\sum_{v\in V}\prod_{{n_j}=\alpha_{\min}}^{\alpha_{\max}}\prod _{S=1}^{s}\chi_u(v)z_{S:00...0}^{\sum_{k=1}^s\eta_{S:00...0}(v^k)}...\; z_{S:i_1...i_{n_j}}^{\sum_{k=1}^s\eta_{S:i_1...i_{n_j}}(v^k)}...\; z_{S:q-1...q-1}^{\sum_{k=1}^s\eta_{S:q-1...q-1}(v^k)}.
\end{equation*}

We rewrite each vector $u$ and $v$ in their $s-$level representation, and we observe that $\chi_u(v)=\chi_{u^1}(v^1)...\chi_{u^s}(v^s).$ Collecting each $v^i$, $1 \leq i \leq s$ under a single sum, we obtain
\begin{eqnarray*}
 \tilde{f}(u)&=&\Big( \sum_{v^1\in R^{n_1}}\chi_{u^1}(v^1)z_{1:00...0}^{\eta_{S:00...0}(v^1)}...z_{1:i_1...i_{n_1}}^{\eta_{S:i_1...i_{n_1}}(v^1)}...\; z_{1:q-1...q-1}^{\eta_{S:q-1...q-1}(v^1)} \Big)\\
&&\Big( \sum_{v^2\in R^{n_2}}\chi_{u^2}(v^2)z_{2:00...0}^{\eta_{S:00...0}(v^2)}...z_{2:i_1...i_{n_2}}^{\eta_{S:i_1...i_{n_2}}(v^2)}...\; z_{2:q-1...q-1}^{\eta_{S:q-1...q-1}(v^2)} \Big)\\
&&...\\
&&\Big( \sum_{v^s\in R^{n_s}}\chi_{u^s}(v^s)z_{s:00...0}^{\eta_{S:00...0}(v^s)}...\; z_{s:i_1...i_{n_s}}^{\eta_{S:i_1...i_{n_s}}(v^s)}...\; z_{s:q-1...q-1}^{\eta_{S:q-1...q-1}(v^s)} \Big).
\end{eqnarray*}
We use the definition of $\chi_u$, (\ref{characterfunction}), and rewrite each $v^i$ more explicitly under the observation of Definition \ref{byteweightdef2}. Hence, we obtain
\begin{eqnarray*}
\tilde{f}(u)&=&\Big( \sum_{\beta_{j_1},...,\beta_{j_{n_1}}\in R^{n_1}}\chi(({\beta_{j_1},...,\beta_{j_{n_1}}})u^1)z_{1:{j_1},...,{j_{n_1}}}\Big)\\
&&\Big( \sum_{\beta_{j_1},...,\beta_{j_{n_2}}\in R^{n_2}}\chi(({\beta_{j_1},...,\beta_{j_{n_2}}})u^2)z_{2:{j_1},...,{j_{n_2}}} \Big)\\
&&...\\
&&\Big( \sum_{\beta_{j_1},...,\beta_{j_{n_s}}\in R^{n_s}}\chi(({\beta_{j_1},...,\beta_{j_{n_s}}})u^s)z_{s:{j_1},...,{j_{n_s}}} \Big).
\end{eqnarray*}

Recalling that $\mu_{S:i_1,i_2,...i_{n_j}}$ counts the number of $s-$levels of the codeword $u$ that are equal to the $n_{j}-$tuple $(\beta_{i_1},\beta_{i_2},...,\beta_{i_{n_j}})$ we obtain

\[\tilde f\left( u \right) = \prod\limits_{{n_j} = a}^b {\prod\limits_{S = 1}^s {\prod\limits_{{i_1},{i_2},...,{i_{{n_j}}} \in Q}^{} A } }, \]
where
\[A = {\left( {\sum\limits_{{\beta _{{j_1}}},{\beta _{{j_2}}},...,{\beta _{{j_{{n_j}}}}} \in {R^{{n_j}}}}^{} {\chi \left( {\left( {{\beta _{{j_1}}},{\beta _{{j_2}}},...,{\beta _{{j_{{n_j}}}}}} \right)\left( {{\beta _{{i_1}}},{\beta _{{i_2}}},...,{\beta _{{i_{{n_j}}}}}} \right)} \right){z_{S:{j_1},{j_2},...,{j_{{n_j}}}}}} } \right)^{{\mu _{S:{i_1},{i_2},...,{i_{{n_j}}}}}\left( u \right)}}.\]

Now by applying the equality (\ref{unluteorem}) given at beginning of the proof, we obtain the desired result.
\end{proof}

\section{Complete poset level weight enumerator of a linear $P_R-$ code}\label{sec4}

In this section, we introduce the definition of complete m-spotty poset level weight enumerator of a linear $P_R-$code over Frobenius rings and obtain a MacWilliams-type identity for the complete poset level weight enumerator.

\begin{definition}\label{definitioncompleteweight}
Let C be a linear $P_R-$code over a poset of size $N$ and $s$ levels. Then the complete poset level weight enumerator of $C$ is defined as follows:
\begin{equation*}
C_W\left( C| z_{i: w(\overline{a})}  :   \overline{a} \in R^{n_i}, 1\leq i\leq s \right) =
\sum\limits_{c \in C} {\prod\limits_{i = 1}^s {{z_{i:{w}\left( {{u^i}} \right)}}} }
\end{equation*}
where $u^i$ denotes the part in the $i$th level of a codeword.
\end{definition}
Let $l_j$ $\left( {0 \le j \le {n_j}} \right)$ be the Hamming weight of the part in the $j$th level of a vector $v,$ and $\left( {{l_1},{l_2},...,{l_s}} \right)$ be the Hamming weight spectrum vector of $v$. Then, we can express the last equality in an equivalent but different form which can be directly expressed by the parameters $l = \left( {{l_1},{l_2},...,{l_s}} \right)$ and ${A_l} = {A_{\left( {{l_1},{l_2},...,{l_s}} \right)}}$ as follows:

\begin{equation*}
C_W\left( C| z_{i: w(\overline{a})}  :   \overline{a} \in R^{n_i}, 1\leq i\leq s \right) = \sum\limits_l  {{A_l }\prod\limits_{i = 1}^{{s}} {{z_{{j:l_j}}}} }
\end{equation*}
where the summation $\sum\limits_l$ takes over all $l$ satisfying $0 \le {l_j} \le {n_j}$ for each $j$ and $A_l$ denotes the number of codewords in $C$ having the Hamming weight spectrum vector $l = \left( {{l_1},{l_2},...,{l_s}} \right)$.

In the following theorem, we derive a MacWilliams-type identity for the complete poset level weight enumerator of a linear $P_R-$code.

\begin{theorem}\label{theocomplete}
Let $C$ be a linear $P_R-$code over a poset of size $N$ and $s$ levels. Then, the relation between the complete poset level weight enumerator of $C$ and its dual is given by
\begin{multline*}
{C_W}\left( \left. {{C^ \bot }} \right|z_{i: w(\overline{a})}  :   \overline{a} \in R^{n_i}, 1\leq i\leq s \right)
\\
= \frac{1}{{\left| C \right|}}\sum\limits_l  {{A_l }\prod\limits_{j= 1}^s {\left( {\sum\limits_{{p_j} = 0}^{{n_j}} {\left\{ {\sum\limits_{{a_j} = 0}^{{p_j}} {{{\left( { - 1} \right)}^{{a_j}}}{{\left( {q - 1} \right)}^{{p_j} - {a_j}}} {\begin{array}{*{20}{c}}
{l_j} \choose {a_j}
\end{array}}  {\begin{array}{*{20}{c}}
{{n_j}-{l_j}} \choose {{p_j}-{a_j}}
\end{array}}} } \right\}{z_{{j:p_j}}}} } \right)} },
\end{multline*}
where  $\chi$ is a nontrivial additive character of $R$ and the summation $\sum\limits_l$ takes over all $l$ satisfying $0 \le {l_j} \le {n_j}$ for each $j$ and $A_l$ denotes the number of codewords in $C$ having the Hamming weight spectrum vector $l = \left( {{l_1},{l_2},...,{l_s}} \right)$. Here, ${{l_j} \choose {a_j}}=0$ for ${l_j}  < {a_j}$ and ${0 \choose 0}=1$.
\end{theorem}

\begin{proof}
Set $f\left( v \right) = \prod\limits_{j= 1}^n {{z_{{w}\left( {{v_j}} \right)}}}$, by Lemma \ref{lemma2}, we obtain

\begin{eqnarray*}
\widetilde f\left( u \right) &=& \sum\limits_{v \in {R^N}} {\chi \left( {\left\langle {u,v} \right\rangle } \right)f\left( v \right)}  = \sum\limits_{v \in {R^N}} {\chi \left( {\left\langle {u,v} \right\rangle } \right)\prod\limits_{j = 1}^s {{z_{{j:w}\left( {{v_j}} \right)}}} } \\
& =& \sum\limits_{{v_1} \in {R^{{n_1}}}} {...\sum\limits_{{v_2} \in {R^{{n_s}}}} {\chi \left( {{u_1}{v_1} + ... + {u_s}{v_s}} \right)\prod\limits_{j = 1}^s {{z_{{j:w}\left( {{v_j}} \right)}}} } } \\
& =& \sum\limits_{{v_1} \in {R^{{n_1}}}} {...\sum\limits_{{v_2} \in {R^{{n_s}}}} {\left( {\prod\limits_{j = 1}^s {\chi \left( {\left\langle {{u_j},{v_j}} \right\rangle } \right){z_{{j:w_H}\left( {{v_j}} \right)}}} } \right)} } \\
 &=& \prod\limits_{j = 1}^s {\left( {\sum\limits_{{v_j} \in {R^{{n_j}}}} {\chi \left( {\left\langle {{u_j},{v_j}} \right\rangle } \right){z_{{j:w}\left( {{v_j}} \right)}}} } \right)}.
\end{eqnarray*}
\end{proof}
\begin{figure}[h!]
\centering
\hfill
  \includegraphics[bb=0 0 500 150 ,scale=0.7]{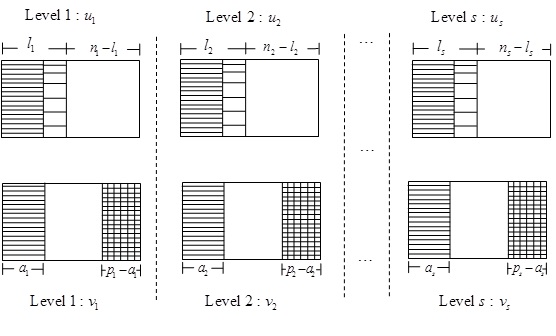}\\
  \caption{The shaded area represents nonzero bits in each level}\label{vedatt}

\end{figure}

Let $v_j$ be a vector in $R^{n_j}$ for $1\leq j\leq s$, having $w(v_j)=p_j$. Assume that the vector $u_j$ and $v_j$ both have nonzero components at $a_j$ positions out of $l_j$ nonzero positions of $u_j$, which can be chosen in $l_j \choose{a_j}$ ways. By Figure \ref{vedatt}, the remaining $p_j-a_j$ nonzero positions of $v_j$ coincide with the zero positions of $u_j$, which are $n_j-l_j$ in number, and can be chosen in ${n_j-l_j}\choose{p_j-a_j}$ ways. Also $w(v_j)=p_j$, the nonzero entries of $v_j$ appear at $p_j$ positions, say at $l_1,l_2,...,l_{p_{j}}$. Then,

\begin{equation*}
\sum\limits_{\begin{array}{*{20}{c}}
  {w\left( {{v_j}} \right) = p_j} \\
  {v_j \in {R^{n_j}}}
\end{array}} {\chi \left( {\left\langle {{u_j},{v_j}} \right\rangle } \right)z_{j:w(v_j)}}  = \sum\limits_{\begin{array}{*{20}{c}}
  {{v_{{l_1}}},{v_{{l_2}}},...,{v_{{l_{p_j}}}} \in {R^ * }} \\
  {v \in {R^b}}
\end{array}} {\chi \left( \sum_{k=1}^j{{u_{{l_k}}}{v_{{l_k}}}} \right)z_{j:p_j}},
\end{equation*}
where ${R^*} = R\backslash \left\{ 0 \right\}.$

Since both $u_j$ and $v_j$ have non-zero bits at $a_j$ positions, assume that $u_i$ has non-zero bits positions $l_1,$ $l_2,$ ..., $l_{a_j}.$ Then, we obtain the following equality:

\begin{equation*}
\begin{gathered}
  \sum\limits_{a_j = 0}^{p_j} {\sum\limits_{\begin{array}{*{20}{c}}
  {w\left( {{v_j}} \right) = a_j} \\
  {v_j \in {R^{n_j}}}
\end{array}} {\chi \left( {\left\langle {{u_j},{v_j}} \right\rangle } \right){z_{j:w(v_j)}}} }  \hfill \\
   \,\,\,\,\,\,\,\,\,= \sum\limits_{a_j = 0}^{p_j} {\left( {\begin{array}{*{20}{c}}
  l_j \\
  {a_j}
\end{array}} \right)\left( {\begin{array}{*{20}{c}}
  {n_j - l_j} \\
  p_j-a_j
\end{array}} \right)\prod\limits_{k = 1}^{a_j} {\left( {\sum\limits_{{v_{{l_k}}} \in {R^*}} {\chi \left( {{u_{{l_k}}}{v_{{l_k}}}} \right)} } \right)} \prod\limits_{k = a_j+1}^{p_j} {\left( {\sum\limits_{{v_{{l_k}}} \in {R^*}} {\chi \left( {0{v_{{l_k}}}} \right)} } \right)} }.  \hfill \\
\end{gathered}
\end{equation*}

Since $\chi$ is a non-trivial generating character, by Corollary
\ref{corol1}, ${\sum\limits_{{v_{{l_k}}} \in {R^*}} {\chi \left(
{{u_{{l_k}}}{v_{{l_k}}}} \right)} }=-1.$ However, using
$\chi(0)=1$ and $\left| {{R^*}} \right| = q - 1$, we obtain the following:

\begin{equation*}
\sum\limits_{\begin{array}{*{20}{c}}
  {w\left( {{v_j}} \right) = p_j} \\
  {v_j \in {R^{n_j}}}
\end{array}} {\chi \left( {\left\langle {{u_j},{v_j}} \right\rangle } \right){z_{j:w(v_j)}}}  = \sum\limits_{a_j = 0}^{p_j} {{{\left( { - 1} \right)}^{a_j}}{{\left( {q - 1} \right)}^{p_j-a_j}}\left( {\begin{array}{*{20}{c}}
  l_j \\
  {a_j}
\end{array}} \right)\left( {\begin{array}{*{20}{c}}
  {n_j - a_j} \\
  p_j-a_j
\end{array}} \right){z_{j:w(v_j)}}}.
\end{equation*}

As $v_j$ runs over the elements of $R^{n_j}$, its Hamming weight $p_j$
varies from $0$ to $n_j.$ Hence,

\begin{equation*}
\sum\limits_{{v_j} \in {R^{n_j}}} {\chi \left( {\left\langle
{{u_j},{v_j}} \right\rangle } \right){z_{j:w(v_j)}}}  = \sum\limits_{p_j = 0}^{n_j} {\left( {\sum\limits_{a_j = 0}^{p_j} {{{\left( { - 1} \right)}^{a_j}}{{\left( {q - 1} \right)}^{p_j-a_j}}\left( {\begin{array}{*{20}{c}}
  l_j \\
  {a_j}
\end{array}} \right)\left( {\begin{array}{*{20}{c}}
  {n_j - l_j} \\
  p_j-a_j
\end{array}} \right)} } \right)} {z_{j:p_j}}.
\end{equation*}

Since $A_l$ is the number of codewords in $C$ having the Hamming weight spectrum vector $l=(l_1,l_2,...,l_s)$, then
\begin{equation}\label{eq80}
\sum\limits_{u \in C} {\widetilde f\left( u \right) = }\sum\limits_l  {{A_l }\prod\limits_{j= 1}^s {\left( {\sum\limits_{{p_j} = 0}^{{n_j}} {\left\{ {\sum\limits_{{a_j} = 0}^{{p_j}} {{{\left( { - 1} \right)}^{{a_j}}}{{\left( {q - 1} \right)}^{{p_j} - {a_j}}} {\begin{array}{*{20}{c}}
{l_j} \choose {a_j}
\end{array}}  {\begin{array}{*{20}{c}}
{{n_j}-{l_j}} \choose {{p_j}-{a_j}}
\end{array}}} } \right\}{z_{{j:p_j}}}} } \right)} },
\end{equation}
where the summation runs over all $s-$tuples.

Substituting Eq.(\ref{eq80}) into Eq.(\ref{eq81}) Lemma \ref{lemma2} we get the result.

\subsection{Poset level weight enumerator of a linear $P_R-$code}

In this subsection, we extend the definition of poset level weight enumerator of linear codes to Frobenius rings, originally given for the binary field in \cite{akbiyik1} and establish a MacWilliams-type identity for these weight enumerators.
\begin{definition}\label{defposetlevel}
Let C be a linear $P_R-$code over a poset of size $N$ and $s$ levels. Then the poset level weight enumerator of $C$ is defined as follows:
\begin{equation}
P_W\left( {C |{z_1},{z_2},...,{z_s}} \right) = \sum\limits_{c \in C} {\prod\limits_{i = 1}^s {z_i^{{w}\left( {{u^i}} \right)}} },
\end{equation}
where $u^i$ denotes the part in the $i$th level of a codeword.
\end{definition}

\begin{example}
Let $P$ be a partial ordered set with order relation $1<2<3$. Consider two binary linear $P_{\mathbb{F}_2}-$codes  given by
\begin{center}
${C_1} = \left\{ {000,001} \right\}\,\,\,\,{\rm{and}}\,\,\,\,{C_2} = \left\{ {000,111} \right\}.$
\end{center}
It is easily seen that the linear $P_{\mathbb{F}_2}-$codes $C_1$ and $C_2$ have the following dual codes, respectively:
\begin{center}
$C_1^ \bot  = \left\{ {000,100,010,110} \right\}\,\,\,\,{\rm{and}}\,\,\,\,C_2^ \bot  = \left\{ {000,110,101,011} \right\}.$
\end{center}

By Definition \ref{defposetlevel}, the poset weight enumerators of these codes and their dual codes are obtained as follows:
\begin{center}
$W\left( {{C_1}| {z_1},{z_2},{z_3}} \right) = 1 + {z_3}\,\,\,\,{\rm{and}}\,\,\,\,W\left( {C_1^ \bot |{z_1},{z_2},{z_3}} \right) = 1 + {z_1} + {z_2} + {z_1}{z_2}$,
\end{center}
\begin{center}
$W\left( {{C_2}|{z_1},{z_2},{z_3}} \right) = 1 + {z_1}{z_2}{z_3}\,\,\,\,{\rm{and}}\,\,\,\,W\left( {C_2^ \bot |{z_1},{z_2},{z_3}} \right) = 1 + {z_1}{z_2} + {z_1}{z_3} + {z_2}{z_3}.$
\end{center}
\end{example}

In the following corollary, we establish a MacWilliams-type identity for $P_W\left( {C| {z_1},{z_2},...{z_s}} \right)$ and $P_W\left( {C^\bot | {z_1},{z_2},...{z_s}} \right)$ as follows:
\begin{corollary}\label{theoposetweight}
\begin{equation*}
{P_W}( C^ \bot)
\\= \frac{1}{{\left| C \right|}}\sum\limits_l  {{A_l }\prod\limits_{j= 1}^s {\left( {\sum\limits_{{p_j} = 0}^{{n_j}} {\left\{ {\sum\limits_{{a_j} = 0}^{{p_j}} {{{\left( { - 1} \right)}^{{a_j}}}{{\left( {q - 1} \right)}^{{p_j} - {a_j}}} {\begin{array}{*{20}{c}}
{l_j} \choose {a_j}
\end{array}}  {\begin{array}{*{20}{c}}
{{n_j}-{l_j}} \choose {{p_j}-{a_j}}
\end{array}}} } \right\}{z_j^{{p_j}}}} } \right)} },
\end{equation*}
\end{corollary}
\begin{proof}
This follows from Theorem \ref{theocomplete} that the poset level weight enumerator of a linear $P_R-$code $C$ can be obtained from the complete poset weight enumerator of $C$ by replacing $z_{j:{p_j}}$ with $z_j^{{p_j}}$ for each $j$.
\end{proof}
\subsection{M-spotty poset level weight enumerator of a linear $P_R-$code}

In this subsection, the results in \cite{umanesan1} are generalized. Further, it is easily seen that the results of this section are special results of the previous section.
\begin{definition}
 A $t_i/n_i$ spotty byte error is defined as $t_i$ or fewer bits errors within a $n_i-$bit  byte, where $1\leq t_i \leq n_i$ for $i \in \left\{ {1,2,..,s} \right\}$.
\end{definition}
Here, if we let $n=n_i$ for all $i$ and take the ring to be the binary field then this definition and the results in \cite{umanesan1} follow as corollaries.

We now introduce the m-spotty poset level weight and the m-spotty poset level distance as follows:

\begin{definition}
Let $e \in R^N$ be an error vector and $e^i \in R^{n_i}$ be the $i$th level of $e$, where $1 \leq i \leq s$. The m-spotty poset level weight, denoted by $w_{MP},$ is defined as
\begin{equation*}
w_{MP}(e)=\sum_{i=1}^s \lceil \frac{w(e^i)}{t_i} \rceil,
\end{equation*}
where $\lceil \frac{w(e^i)}{t_i} \rceil$ denotes the ceiling of $\frac{w(e^i)}{t_i}$, i.e. the smallest integer which is not less than $\frac{w(e^i)}{t_i}$.
\end{definition}

\begin{definition}
Let $u$ and $v$ be codewords of a linear $P_R-$code $C$. Then m-spotty poset level distance between $u$ and $v$, denoted by $d_{MP}(u,v)$, is defined as follows:\\
\begin{equation*}
d_{MP}(u,v)=\sum_{i=1}^s \lceil \frac{d(u^i,v^i)}{t_i} \rceil.
\end{equation*}
Here $u^i$ and $v^i$ denote the $i$th level of $u$ and $v$, respectively.
\end{definition}
\begin{theorem}
The m-spotty poset level distance is a metric over $R$.
\end{theorem}

\begin{proof}
It is easy to see that $d_{MP}(u,v) \ge 0$ for $u\neq v$, $d_{MP}(u,v)=0$ for $u=v$ and $d_{MP}(u,v)=d_{MP}(v,u)$. So, we only need to show that the triangle equality, i.e. $d_{MP}(u,v)\leq d_{MP}(u,w)+d_{MP}(w,v)$ for every $u$, $v$ and $w$ $\in R^N$. Since the Hamming distance function is a metric, then $d(u^i,v^i)\leq d(u^i,w^i)+d(w^i,v^i)$ can be written for $i=1,2,...,s$. So $\frac{d(u^i,v^i)}{t_i}\leq \frac{d(u^i,w^i)+d(w^i,v^i)}{t_i}$. For each level\\
\[\lceil\frac{d(u^1,v^1)}{t_1}\rceil \leq \lceil\frac{d(u^1,w^1)+d(w^1,v^1)}{t_1}\rceil \leq \lceil\frac{d(u^1,w^1)}{t_1}\rceil + \lceil\frac{d(w^1,v^1)}{t_1}\rceil, \]\\
\[\lceil\frac{d(u^2,v^2)}{t_2}\rceil \leq \lceil\frac{d(u^2,w^2)+d(w^2,v^2)}{t_2}\rceil \leq \lceil\frac{d(u^2,w^2)}{t_2}\rceil + \lceil\frac{d(w^2,v^2)}{t_2}\rceil, \]
\[.\]
\[.\]
\[.\]
\[\lceil\frac{d(u^s,v^s)}{t_s}\rceil \leq \lceil\frac{d(u^s,w^s)}{t_s}\rceil + \lceil\frac{d(w^s,v^s)}{t_s}\rceil. \]
By summing all the inequalities from $i=1$ to $n$, we see that $d_{MP}(u,v)\leq d_{MP}(u,w)+d_{MP}(w,v)$.
\end{proof}

\begin{definition}
The weight enumerator for m-spotty byte error control code $C$ is defined as:
\begin{equation*}
M_W(C|z_1,z_2,...,z_s)=\sum_{u\in C}z^{w_{MP}(u)}=\sum_{u\in C}\prod _{i=1}^sz_i^{w_{MP}(u^i)}.
\end{equation*}
\end{definition}

The following theorem gives a relation between the m-spotty poset level weight enumerator of a linear $P_R-$code and that of its dual.

\begin{corollary}\label{theomspottyposet}
\begin{equation*}
{M_W}( C^ \bot)
\\= \frac{1}{{\left| C \right|}}\sum\limits_l  {{A_l }\prod\limits_{j= 1}^s {\left( {\sum\limits_{{p_j} = 0}^{{n_j}} {\left\{ {\sum\limits_{{a_j} = 0}^{{p_j}} {{{\left( { - 1} \right)}^{{a_j}}}{{\left( {q - 1} \right)}^{{p_j} - {a_j}}} {\begin{array}{*{20}{c}}
{l_j} \choose {a_j}
\end{array}}  {\begin{array}{*{20}{c}}
{{n_j}-{l_j}} \choose {{p_j}-{a_j}}
\end{array}}} } \right\}{z_j^{{\left\lceil {\frac{{{p_j}}}{{{t_j}}}} \right\rceil }}}} } \right)} },
\end{equation*}
\end{corollary}
\begin{proof}
This follows from Theorem \ref{theocomplete} that the poset level weight enumerator of a linear $P_R-$code $C$ can be obtained from the complete poset weight enumerator of $C$ by replacing $z_{j:{p_j}}$ with $z_j^{\left\lceil {{{{p_j}} \mathord{\left/
 {\vphantom {{{p_j}} {{t_j}}}} \right.
 \kern-\nulldelimiterspace} {{t_j}}}} \right\rceil }$ for each $j$.
\end{proof}

\section{Examples}\label{sec5}

\begin{figure}[h!]
\centering
\hfill
  \includegraphics[bb=0 0 300 150 ,scale=0.9]{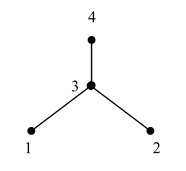}\\
  \caption{A poset of size 4 and 3 levels}\label{noyan}
\end{figure}

In this section, we present  some illustrative examples for byte poset level weight enumerators, complete poset level weight enumerators, poset level weight enumerators and m-spotty poset level weight enumerators over  a special poset given in  Figure \ref{noyan}. Further, their dual weight enumerators are presented as well.

\begin{example}\label{example200}
Let $C=\{0000,1010,0111,1101\}$ be a linear $P_{\mathbb{F}_2}-$code. It can be easily seen that ${C^ \bot } = \left\{ {0000,1011,0101,1110} \right\}$. By Definition \ref{byteweightdef3}, the byte weight enumerators of these codes are as follows:
\begin{eqnarray*}
& B_W(C|z_{1:00},z_{2:00},z_{3:00},z_{1:01},z_{2:01},z_{3:01},z_{1:10},z_{2:10},z_{3:10},z_{1:11},z_{2:11},z_{3:11},z_{1:0},z_{2:0},z_{3:0},z_{1:1},z_{2:1},z_{3:1})\\
& =z_{1:00}z_{2:0}z_{3:0}+z_{1:10}z_{2:1}z_{3:0}+z_{1:10}z_{2:1}z_{3:1}+z_{1:11}z_{2:0}z_{3:1},
\end{eqnarray*}

and by applying Theorem \ref{byteweighttheo1}, we immediately obtain the byte poset level weight enumerator of $C^\bot$ as follows:

\begin{eqnarray*}
& B_W(C^{\perp}|z_{1:00},z_{2:00},z_{3:00},z_{1:01},z_{2:01},z_{3:01},z_{1:10},z_{2:10},z_{3:10},z_{1:11},z_{2:11},z_{3:11},z_{1:0},z_{2:0},z_{3:0},z_{1:1},z_{2:1},z_{3:1})\\
& =z_{1:00}z_{2:0}z_{3:0}+z_{1:10}z_{2:1}z_{3:1}+z_{1:01}z_{2:0}z_{3:1}+z_{1:11}z_{2:1}z_{3:0}.
\end{eqnarray*}
\end{example}

\begin{example}\label{example201}
Consider again the linear $P_{\mathbb{F}_2}-$code $C$ in Example \ref{example200}. By Definition \ref{definitioncompleteweight}, we compute the complete weight enumerators of these codes as follows:

\begin{eqnarray*}
&C_W(C|z_{1:0},z_{2:0},z_{3:0},z_{1:1},z_{2:1},z_{3:1},z_{1:2})\\
&=z_{\textbf{1}:0}z_{\textbf{2}:0}z_{\textbf{3}:0}+z_{\textbf{1}:1}z_{\textbf{2}:1}z_{\textbf{3}:0}+z_{\textbf{1}:1}z_{\textbf{2}:1}z_{\textbf{3}:1}+z_{\textbf{1}:2}z_{\textbf{2}:0}z_{\textbf{3}:1},
\end{eqnarray*}

and by Theorem \ref{theocomplete}, we obtain the complete poset weight enumerator of $C^\bot$ in the above equality, that is,
\begin{eqnarray*}
&C_W(C^{\perp}|z_{1:0},z_{2:0},z_{3:0},z_{1:1},z_{2:1},z_{3:1},z_{1:2})\\
&=z_{\textbf{1}:0}z_{\textbf{2}:0}z_{\textbf{3}:0}+z_{\textbf{1}:1}z_{\textbf{2}:1}z_{\textbf{3}:1}+z_{\textbf{1}:1}z_{\textbf{2}:0} z_{\textbf{3}:1}+z_{\textbf{1}:2}z_{\textbf{2}:1}z_{\textbf{3}:0}.
\end{eqnarray*}
\end{example}

\begin{example}
By Corollary \ref{theoposetweight} and Corollary \ref{theomspottyposet}, we replace $z_{j:k}$ with $z_j^k$ and $z_j^k$ with $z_j^{\left\lceil {{k \mathord{\left/
 {\vphantom {k {{t_j}}}} \right.
 \kern-\nulldelimiterspace} {{t_j}}}} \right\rceil }$ in Example \ref{example201}, respectively, where $1 \leq j \leq 3,$ $t=(t_1,t_2,t_3)=(2,1,1)$. Therefore, the poset weight enumerator and the m-spotty poset level weight enumerator of $C^\bot$ via the complete poset weight enumerator of $C^\bot$ given in Example \ref{example201} can be easily seen as follows:
\begin{eqnarray*}
&P_W(C^{\perp}|z_1, z_2, z_3)={z_\textbf{1}}^0{z_\textbf{2}}^0{z_\textbf{3}}^0+{z_\textbf{1}}^1{z_\textbf{2}}^1{z_\textbf{3}}^1+{z_\textbf{1}}^1{z_\textbf{2}}^0 {z_\textbf{3}}^1+{z_\textbf{1}}^2{z_\textbf{2}}^1{z_\textbf{3}}^0\\
&=1+z_1z_2z_3+z_1z_3+z_1^2z_2,
\end{eqnarray*}
and
\begin{equation*}
M_W(C^{\perp}|z_1,z_2,z_3)=1+z_1z_2z_3+z_1z_3+z_1z_2.
\end{equation*}
\end{example}

\section{Conclusion}\label{sec6}
In this study, byte poset level weight enumerators, complete level weight enumerators, poset level weight enumerators and $m-$spotty poset level weight enumerators of  linear poset codes over Frobenius rings are introduced. Respectively their MacWilliams identities are established.  Naturally this contribution has led to generalizations of some recent work done on weight enumerators and due to this generalizations some well-known weight enumerators and their MacWilliams type identities are easily obtained as corollaries. Since these are new weight enumerators, similar to the classical cases there are many new directions for further studying. For instance, Gleason's type theorems and lattices related to these weight enumerators are just a few to mention here.

\bibliographystyle{amsplain}

\begin{thebibliography}{40}
\bibitem{1}  R. Ahlswede, L. A. Bassalygo, and M. S. Pinsker, Nonbinary codes correcting localized errors, IEEE Trans. Inform. Theory. 39 (1993) 1413-1416.

\bibitem{akbiyik1} S. Akbiyik, I. Siap, A P-complete weight enumerator with respect to poset metric and its MacWilliams identity (Turkish), Adiyaman University J. Sci. 1 (2011) 28-39.

\bibitem{akbiyik2}S. Akbiyik, I. Siap, 	MacWilliams Identities Over Some Special Posets,  Communications Faculty of Sciences University of Ankara, Series   A1: Mathematics and Statistics, Vol.  62 (1) (2013), in print.

\bibitem{2}  L. A. Bassalygo, S. I. Gelfand, M. S. Pinsker, Coding for channels with localized errors, Proc. 4th Soviet-Swedish Workshop in
Inform. Theory, Sweden, 1989, pp. 95-99.

\bibitem{3} L. A. Bassalygo, S. I. Gelfand, M. S. Pinsker, Coding for partially localized errors, IEEE Trans. Inform. Theory. 37 (1991) 880-884.

\bibitem{brualdi} R. A. Brualdi, J. Graves, K. M. Lawrence, Codes with a poset metric, Discrete Math. 147 (1995) 57-72.

\bibitem{chang}  S. H. Chang, M. Rim, P. C. Cosman, L. B. Milstein, Optimized unequal error protection using multiplexed hierarchical modulation, IEEE Trans. Inform. Theory. 58 (2012) 5816-5840.

\bibitem{Claasen} H. L. Claasen, R. W. Goldbach, A field-like property of finite field, Indag. Math. 3 (1992) 11-26.

\bibitem{choi} S. Choi, J. Y. Hyun, D. Y. Oh, H. K. Kim, MacWilliams-type equivalence relations, available in arXiv:1205.1090v2.

\bibitem{felix} L. V. Felix, M. Firer, Canonical-systematic form for codes in hierarchical poset metrics, Adv. Math. Commun. 6 (2012) 315--328.

\bibitem{posetdecoders} M. Firer, L. Panek, L. Rifo,
Coding in the presence of semantic value of information: Unequal error protection using poset decoders, available in arXiv:1108.3832v1  [cs.IT], Partial version in AIP Conf. Proc. 1490: 126-134.

\bibitem{fujiwara} E. Fujiwara, Code Design for Dependable Systems: Theory and Practical Applications, A John Wiley \& Sons Inc. Pub., New Jersey, 2006.

\bibitem{article. 7} J. N. Gutierrez, H. Tapia-Recillas, A MacWilliams
Identity for Poset-Codes, Congr. Numer. 133 (1998) 63--73.

\bibitem{huang} K. Huang, C. Liang, X. Ma, B. Bai, Unequal error protection by partial superposition transmission using LDPC Codes. Available in  arXiv:1309.3864 [cs.IT], 2013.

\bibitem{article.2} H. K.Kim, D. Y. Oh , A Classification of Posets Admitting the
MacWilliams Identity, IEEE Transactions on Information Theory, 51 (4) (2005) 1424--1431.

\bibitem{Packet} J. Korhonen, P. Frossard, Bit-error resilient packetization for streaming H.264/AVC video,  Proc. of the 1st ACM International Workshop on Mobile Video, Germany, 2007, pp. 25-30.

\bibitem{kuriata}  E. Kuriata, Creation of unequal error protection codes for  two groups of symbols,Int. J. Appl. Math. Comput. Sci.  18 (2008) 251-257.


\bibitem{kwang} H. K.Kim, D. Y. Oh, A classification of posets admitting the MacWilliams identity, IEEE Trans. Inform. Theory. 51 (2005) 1424-1431.

\bibitem{4}  P. Larsson, Codes for Correction of Localized Errors (Linkoping Studies in Science and Technology, dissertations, no. 374). Linkoping Sweden, (1995).

\bibitem{macwilliams} F. J. MacWilliams, N. J. A. Sloane, The Theory of Error Correcting Codes, North-Holland Pub. Co., Amsterdam, 1977.

\bibitem{0}A. Mardjuadi, J. H. Weber, Codes for multiple localized burst error correction, IEEE Trans. Inform. Theory. 44 (1998) 2020-2024.

\bibitem{masnik}  B. Masnik , J. Wolf,  On linear unequal error protection
codes, IEEE Trans. on Information Theory, IT-
13(4) (1967) 600–607.

\bibitem{Niederreiter1} H. Niederreiter, Point sets and sequences with small discrepancy, Mh. Math. 104 (1987) 273-337.

\bibitem{Niederreiter2}H. Niederreiter, A combinatorial problem for vector spaces over finite fields, Discrete Math. 96 (1991) 221-228.

\bibitem{Niederreiter3}H. Niederreiter, Orthogonal arrays and other combinatorial aspects in the theory of uniform point distributions in unit cubes, Discrete Math. 106/107 (1992) 361-367.

\bibitem{vedatsiap1} M. Ozen, V. Siap, The MacWilliams identity for m-spotty weight enumerators of linear codes over finite fields, Comput. Math. Appl. 61 (2011) 1000-1004.

\bibitem{vedatsiap3} M. Ozen, V. Siap, The MacWilliams identity for m-spotty Rosenbloom-Tsfasman weight enumerator, J. Franklin Inst. doi:10.1016/j.franklin.2012.06.002.



\bibitem{hier} J. A. Pinheiro, M. Firer, Classification of Poset-Block Spaces Admitting MacWilliams-Type Identity, IEEE Trans. Inform. Theory. 58 12 (2012) 7246 - 7252.



\bibitem{evet} R. M. Roth, G. Seroussi, Location-correcting codes, IEEE Trans. Inform. Theory. 42 (1996) 554-565.

\bibitem{sharma1} A. Sharma, A. K. Sharma, MacWilliams type identities for some new m-spotty weight enumerators, IEEE Trans. Inform. Theory.  58 (2012) 3912-3924.

\bibitem{sharma2} A. Sharma, A. K. Sharma, On some new m-spotty Lee weight enumerators, Des. Codes Cryptogr. doi: 10.1007/s10623-012-9725-z.

\bibitem{sharma3} A. Sharma, A. K. Sharma, On MacWilliams type identities for r-fold joint m-spotty weight enumerators, Discrete Math. 312 (2012) 3316-3327.

\bibitem{irfansiap5} I. Siap, The complete weight enumerator for codes over $M_{m\times s}(F_q)$, Lect. Notes Comput. Sci. (2001) 20-26.

\bibitem{irfansiap4} I. Siap, A MacWilliams type identity, Turk. J. Math. 26  (2002) 465-473.

\bibitem{irfansiap3} I. Siap, M. Ozen, The complete weight enumerator for codes over $M_{n\times s}(R)$, Appl. Math. Lett. 17 (2004) 65-69.

\bibitem{irfansiap1} I. Siap, MacWilliams identity for m-spotty Lee weight enumerator, Appl. Math. Lett. 23 (2010) 13-16.

\bibitem{irfansiap2} I. Siap, An identity between the m-spotty weight enumerators of a linear code and its dual, Turk. J. Math. 36 (2012) 641-650.

\bibitem{vedatsiap2} V. Siap, M. Ozen, MacWilliams identity for m-spotty Hamming weight enumerator over the ring $F_2+vF_2$, Eur. J. Pure Appl. Math. 5 (2012) 373-379.

\bibitem{vedatsiap4} V. Siap, A MacWilliams type identity for m-spotty generalized Lee weight enumerators over $Z_q$, Math. Sci. and Appl. E-Notes. 1 (2013) 111-116.

\bibitem{simonis} J. Simonis,  MacWilliams  identities  and  coordinate  partitions, Linear Alg. Appl. 216 (1995) 81-91.

\bibitem{suzuki1} K. Suzuki, T. Kashiyama, E. Fujiwara, A general class of m-spotty byte error control codes, IEICE Trans. Fundam. E90-A (2007) 1418-1427.

\bibitem{suzuki2} K. Suzuki, E. Fujiwara, MacWilliams identity for m-spotty weight enumerator, IEICE Trans. Fundam. E93-A (2010) 526-531.

\bibitem{suzuki3} K. Suzuki, Complete m-spotty weight enumerators of binary codes, Jacobi forms, and partial Epstein zeta functions, Discrete Math. 312 (2012) 265-278.

\bibitem{umanesan1} G. Umanesan, E. Fujiwara, A class of random multiple bits in a byte error correcting and single byte error detecting (S$_{t/b}$EC-S$b$ED) codes, IEEE Trans. Comput. 52 (2003) 835-847.

\bibitem{umanesan2} G. Umanesan, E. Fujiwara, A class of codes for correcting single spotty byte errors, IEICE Trans. Fundam. E86-A (2003) 704-714.

\bibitem{morelos} R. M. Zaragoza, M. Fossorier, S. Lin, H. Imai, Multilevel coded modulation for unequal error protection and multistage decoding-Part I: Symmetric constellations, IEEE Trans. Commun. 48 (2000) 204-213.



\bibitem{wood} J. Wood, Duality for modules over finite rings and applications to coding theory, Amer. J. Math., (1999) 555-575.
\end{thebibliography}

\end{document}